\documentclass[11pt]{article}

\usepackage{epsf}
\usepackage{latexsym}
\usepackage{amssymb}
\usepackage{amsmath}
\usepackage[usenames,dvipsnames]{xcolor}
\usepackage{tabularx,ragged2e,booktabs,caption}

\usepackage{stmaryrd}
\usepackage[bottom]{footmisc}

\usepackage{verbatim}

\usepackage{fullpage}
\usepackage{amssymb}
\usepackage{amsmath}
\usepackage{amsopn}
\usepackage{graphics}
\usepackage{epsfig}
\usepackage{color}
\usepackage{enumerate}
\usepackage{cite}
\usepackage{graphicx}
\usepackage{version}
\usepackage[lined,boxed,commentsnumbered]{algorithm2e}

\usepackage{times}

\newtheorem{theorem}{Theorem}[section]
\newtheorem{lemma}[theorem]{Lemma}
\newtheorem{definition}[theorem]{Definition}
\newtheorem{corollary}[theorem]{Corollary}
\newtheorem{claim}[theorem]{Claim}

\newtheorem{property}[theorem]{Property}
\newtheorem{construction}[theorem]{Construction}
\newtheorem{remark}[theorem]{Remark}

\newcommand{\poly}{\mathsf{poly}}
\newcommand{\supp}{\mathsf{supp}}

\DeclareMathOperator{\E}{E}

\def\time_h{3}
\newcommand{\ConstHeight}{10d \cdot \frac{m}{n}+1}

\newcommand{\sq}{\hbox{\rlap{$\sqcap$}$\sqcup$}}
\newcommand{\qed}{\hspace*{\fill}\sq}
\newenvironment{proof}{\noindent {\it Proof.}\ }{\qed\par\vskip 4mm\par}
\newenvironment{proofof}[1]{\bigskip \noindent {\it Proof of #1.}\quad }
{\qed\par\vskip 4mm\par}

\begin{document}

\title{Derandomized Balanced Allocation}
\author{Xue Chen\thanks{Supported by NSF Grant CCF-1526952 and a Simons Investigator Award (\#409864, David Zuckerman), part of this work was done while the author was visiting the Simons Institute.}
\\
Computer Science Department\\
University of Texas at Austin\\
{\tt xchen@cs.utexas.edu}
}

\maketitle

\begin{abstract}
\normalsize
In this paper, we study the maximum loads of explicit hash families in the $d$-choice schemes when allocating sequentially $n$ balls into $n$ bins. We consider the \emph{Uniform-Greedy} scheme \cite{ABKU}, which provides $d$ independent bins for each ball and places the ball into the bin with the least load, and its non-uniform variant --- the \emph{Always-Go-Left} scheme introduced by V\"ocking~\cite{Vocking}. We construct a hash family with $O(\log n \log \log n)$ random bits based on the previous work of Celis et al.~\cite{CRSW} and show the following results.

\begin{enumerate}
\item With high probability, this hash family has a maximum load of $\frac{\log \log n}{\log d} + O(1)$ in the \emph{Uniform-Greedy} scheme.

\item With high probability, it has a maximum load of $\frac{\log \log n}{d \log \phi_d} + O(1)$ in the \emph{Always-Go-Left} scheme for a constant $\phi_d>1.61$.
\end{enumerate}

The maximum loads of our hash family match the maximum loads of a perfectly random hash function \cite{ABKU,Vocking} in the \emph{Uniform-Greedy} and \emph{Always-Go-Left} scheme separately, up to the low order term of constants. Previously, the best known hash families matching the same maximum loads of a perfectly random hash function in $d$-choice schemes were $O(\log n)$-wise independent functions \cite{Vocking}, which needs $\Theta(\log^2 n)$ random bits.
\end{abstract}

\thispagestyle{empty}
\setcounter{page}{0}
\pagebreak

\section{Introduction}
We investigate explicit constructions of hash functions for the classical problem of placing balls into bins.  The basic model is to hash $n$ balls into $n$ bins independently and uniformly at random, which we call $1$-choice scheme. For the $1$-choice scheme, it is well-known that each bin contains at most $O(\frac{\log n}{\log \log n})$ balls with high probability. For convenience, we always use logarithm of base 2 in this work. Here, by high probability, we mean probability $1-n^{-c}$ for an arbitrary constant $c$.

An alternative variant, which we call \emph{Uniform-Greedy}, is to provide $d \ge 2$ independent random choices for each ball and place the ball in the bin with the lowest load. In a seminal work, Azar et al. \cite{ABKU} showed that the \emph{Uniform-Greedy} scheme with $d$ independent random choices guarantees a maximum load of only $\frac{\log \log n}{\log d}+O(1)$ with high probability for $n$ balls, which significantly improved the load-balancing of 1-choice scheme even for $d=2$. This scheme has various  applications in many areas of computer science such as cryptography \cite{ANSS}. We refer to surveys \cite{MRS,M01,Wieder}) for a detailed discussion.

Surprisingly, V\"ocking \cite{Vocking} used an  \emph{asymmetric} scheme, called \emph{Always-Go-Left}, to further improve the maximum load to $\frac{\log \log n}{d \log \phi_d} + O(1)$ for $d \ge 2$ choices, where $\phi_d>1.61$ is the constant satisfying $\phi_d^d=1+\phi_d+\cdots+\phi_d^{d-1}$. This scheme partitions the $n$ bins into $d$ groups with equal size and uses an unfair tie-breaking mechanism such that the allocation process provides $d$ independent choices from the $d$ groups separately for each ball and always allocates it to the left-most bin with the least load. Moreover, V\"ocking \cite{Vocking} showed this load balancing is optimal for any random scheme using $d$ choices. For convenience, we always use $d$-choice schemes to denote both the \emph{Uniform-Greedy} and \emph{Always-Go-Left} scheme with $d \ge 2$ choices in this work.

Traditional analysis of load balancing assumes a perfectly random hash function. A large body of research is dedicated to the removal of this assumption by designing explicit hash families using fewer random bits. In the $1$-choice scheme, it is well known that $O(\frac{\log n}{\log \log n})$-wise independent functions guarantee a maximum load of $O(\frac{\log n}{\log \log n})$ with high probability, which reduces the number of random bits to $O(\frac{\log^2 n}{\log \log n})$. Recently, Celis et al. \cite{CRSW} designed a hash family with a description of $O(\log n \log \log n)$ random bits that achieves the same maximum load of $O(\frac{\log n}{\log \log n})$ as a perfectly random hash function.

In this work, we are interested in the explicit constructions of hash families that achieve \emph{the same maximum loads} as a perfectly random hash function in the $d$-choice schemes, i.e., $\frac{\log \log n}{\log d}+O(1)$ in the \emph{Uniform-Greedy} scheme\cite{ABKU,Vocking} and $\frac{\log \log n}{d \log \phi_d}+O(1)$ in the \emph{Always-Go-Left} scheme\cite{Vocking}. For these two schemes, $O(\log n)$-wise independent hash functions achieve the same maximum loads from V\"ocking's argument \cite{Vocking}, which provides a hash family with $\Theta(\log^2 n)$ random bits. Recently, Reingold et al. \cite{RRW14} showed that the hash family designed by Celis et al. \cite{CRSW} guarantees a maximum load of $O(\log \log n)$ of $n$ balls in the \emph{Uniform-Greedy} scheme with $O(\log n \log \log n)$ random bits.

\subsection{Our Contribution}
For multiple-choice schemes, we strengthen the hash family of Celis et al. \cite{CRSW} for the 1-choice scheme  --- our hash family is $O(\log \log n)$-wise independent over $n$ bins and ``almost" $O(\log n)$-wise independent over a fraction of $\poly(\log n)$ bins. Then we prove that our hash family derandomizes V\"ocking's witness tree argument \cite{Vocking} such that $O(\log n \log \log n)$ random bits could guarantee the same maximum loads as a perfectly random hash function in the multiple-choice schemes.

We first show our hash family guarantees a maximum load of $\frac{\log \log n}{\log d}+O(1)$ in the \emph{Uniform-Greedy} scheme\cite{ABKU,Vocking} with $d$ choices. We use $U$ to denote the pool of balls and consider placing $m=O(n)$ balls into $n$ bins here. Without loss of generality, we always assume $|U|=\poly(n)$ and $d$ is a constant at least $2$ in this work.
\begin{theorem}[Informal version of Theorem \ref{thm:sym}]\label{thm:intro_uniform}
For any $m=O(n)$, any constants $c$ and $d$, there exists a hash family with $O(\log n \log \log n)$ random bits such that given any $m$ balls in $U$, with probability at least $1-n^{-c}$, the max-load of the \emph{Uniform-Greedy} scheme with $d$ independent choices of $h$ is $\frac{\log \log n}{\log d} + O(1)$.
\end{theorem}

Then we show this hash family guarantees a load balancing of $\frac{\log \log n}{d \log \phi_d}+O(1)$ in the \emph{Always-Go-Left} scheme \cite{Vocking} with $d$ choices. Notice that the constant $\phi_d$ in equation $\phi_d^d=1+\phi_d+\cdots+\phi_d^{d-1}$ satisfies $1.61<\phi_2<\phi_3<\phi_4<\cdots<\phi_d<2$. Compared to the \emph{Uniform-Greedy} scheme, the \emph{Always-Go-Left} scheme \cite{Vocking} improves the load exponentially with regard to $d$. Even for $d=2$, the \emph{Always-Go-Left} scheme reduces the load from $\log \log n + O(1)$ to $0.7 \log \log n+O(1)$. From the lower bound $\frac{\log \log n}{d \log \phi_d}-O(1)$ on the load balancing of any random $d$-choice scheme shown by V\"ocking~\cite{Vocking}, the load of our hash family in the \emph{Always-Go-Left} scheme is optimal among $d$-choice schemes up to the low order term of constants.
\begin{theorem}[Informal version of Theorem \ref{thm:goleft}]\label{thm:intro_left}
For any $m=O(n)$, any constants $c$ and $d$, there exists a hash family with $O(\log n \log \log n)$ random bits such that given any $m$ balls in $U$, with probability at least $1-n^{-c}$, the max-load of the \emph{Always-Go-Left} scheme with $d$ independent choices of $h$ is $\frac{\log \log n}{d \log \phi_d} + O(1)$.
\end{theorem}
At the same time, Our hash family has an evaluation time $O\big((\log \log n)^4\big)$ in the RAM model based on the algorithm designed by Meka et al. \cite{MRRR} for the hash family of Celis et al. \cite{CRSW}.

Finally, we show our hash family guarantees the same maximum load as a perfectly random hash function in the $1$-choice scheme for $m=n \cdot \poly(\log n)$ balls. Given $m>n \log n$ balls in $U$, the maximum load  of the $1$-choice scheme becomes $\frac{m}{n} + O(\sqrt{\log n \cdot  \frac{m}{n}})$ from the Chernoff bound. For convenience, we refer to this case of $m \ge n \log n$ balls as a \emph{heavy load}. In a recent breakthrough, Gopalan, Kane, and Meka~\cite{GKM15} designed a pseudorandom generator of seed length $O(\log n (\log \log n)^2)$ that fools the Chernoff bound within polynomial error. Hence the pseudorandom generator \cite{GKM15} provides a hash function with $O(\log n (\log \log n)^2)$ random bits for the heavy load case. Compared to the hash function of~\cite{GKM15}, we provide a simplified construction that achieves the same maximum load but only works for $m=n \cdot \poly(\log n)$ balls.

\begin{theorem}[Informal version of Theorem \ref{thm:heavy_load}]\label{thm:intro_heavy_load}
For any constants $c$ and $a \ge 1$, there exist a hash function generated by $O(\log n \log \log n)$ random bits such that for any $m=\log^a n \cdot n$ balls, with probability at least $1-n^{-c}$, the max-load of the $n$ bins in the 1-choice scheme with $h$ is $\frac{m}{n}+O\left( \sqrt{\log n} \cdot \sqrt{\frac{m}{n}}\right)$.
\end{theorem}

\subsection{Previous Work}
\paragraph{The 1-choice scheme.} Natural explicit constructions of hash functions using a few random bits are $k$-wise independent functions, small-biased spaces, and $k$-wise small-biased spaces. For the $1$-choice scheme with $m=n$ balls, Alon et al.\cite{ADMPT} showed the existence of a pairwise independent hash family that always has a maximum load of $\sqrt{n}$. On the other hand, it is well known that $O(\frac{\log n}{\log \log n})$-wise independent functions achieve a maximum load of $O(\frac{\log n}{\log \log n})$ with high probability, which needs $\Theta(\frac{\log^2 n}{\log \log n})$ random bits. Using $O(\log n)$-wise small-biased spaces as milder restrictions, Celis et al.\cite{CRSW} designed a hash family with $O(\log n \log \log n)$ random bits achieving the same maximum load with high probability. 

For the heavy load case in the $1$-choice scheme, a perfectly random hash function guarantees a maximum load of $\frac{m}{n}+O(\sqrt{\log n \cdot \frac{m}{n}})$ from the Chernoff bound. Hence any pseudorandom generator fooling the Chernoff bound within polynomial small error is a hash family matching this maximum load. Schmidt et al. \cite{SSS95} showed that $O(\log n)$-wise independence could derandomize the Chernoff bound, which provides a hash function with $O(\log^2 n)$ random bits. In a recent breakthrough~\cite{GKM15}, Gopolan, Kane, and Meka designed a pseudorandom generator with seed length $O(\log n (\log \log n)^2)$ to fool halfspaces, the Chernoff bound, and many other classes, which provides a hash family of $O(\log n (\log \log n)^2)$ bits.

\paragraph{Multiple-choice schemes.} For $m=n$ balls in the $d$-choice schemes, the original argument of \cite{ABKU} adopts an inductive proof that relies on the assumption of full randomness. It is folklore (e.g., \cite{RRW14,DKRT}) that $O(\log n)$-wise independent functions could derandomize V\"{o}cking's witness tree argument \cite{Vocking} to achieve a maximum load of $\frac{\log \log n}{\log d}+O(1)$ in the \emph{Uniform-Greedy} scheme, which takes $\Theta(\log^2 n)$ random bits. Reingold et al. \cite{RRW14} prove the hash family of \cite{CRSW} derandomizes Cuckoo hashing and achieves a maximum load of $\log \log n + O(1)$ in the \emph{Uniform-Greedy} scheme for \emph{$n/C_0$} balls for any constant $C_0>2$. This leaves a gap for $m>n$ balls and $d>2$ choices, where the guarantee of the maximum load in Reingold et al. \cite{RRW14} becomes $\frac{C_0 m}{n} \cdot \big( \log \log n + O(1) \big)$.

V\"ocking~\cite{Vocking} introduced \emph{Always-Go-Left} scheme to further improve the maximum loads of $d$-choice schemes to $\frac{\log \log n}{d \log \phi_d}+O(1)$. In the same work, V\"ocking showed a lower bound to illustrate that the load $\frac{\log \log n}{d \log \phi_d}$ is optimal for random $d$-choice schemes. However, much less is known about the derandomization of the \emph{Always-Go-Left} scheme except $O(\log n)$-wise independent functions for V\"{o}cking's witness tree argument \cite{Vocking}, which is pointed out in \cite{RRW14,DKRT}.

We summarize these results in Table \ref{tab:previous_work}.

\begin{minipage}{\linewidth}
\centering
\vspace{0.15in}
\begin{tabular}{|l|c|c|c|}
\hline
\bf scheme & \bf reference & \bf maximum load & \bf number of random bits \\ \hline
1-choice & well known & $O(\frac{\log n}{\log \log n})$ & $\Theta(\frac{\log^2 n}{\log \log n})$ \\
1-choice & \cite{CRSW} & $O(\frac{\log n}{\log \log n})$ & $O(\log n \log \log n)$ \\
\emph{Uniform-Greedy} & \cite{ABKU} & $\frac{\log \log n}{\log d}+O(1)$ & full randomness \\
\emph{Uniform-Greedy} & \cite{Vocking} & $\frac{\log \log n}{\log d}+O(1)$ & $\Theta(\log^2 n)$ \\
\emph{Uniform-Greedy} & \cite{RRW14} & $O(\log \log n)$ & $O(\log n \log \log n)$ \\
\emph{Uniform-Greedy} & this work & $\frac{\log \log n}{\log d}+O(1)$ & $O(\log n \log \log n)$\\ \hline
\emph{Always-Go-Left} & \cite{Vocking} & $\frac{\log \log n}{d \log \phi_d}+O(1)$ & $\Theta(\log^2 n)$\\ \hline
\emph{Always-Go-Left} & this work & $\frac{\log \log n}{d \log \phi_d}+O(1)$ & $O(\log n \log \log n)$\\
\hline
\end{tabular}
\captionof{table}{Previous work about the random bits and maximum loads of multiple-choice schemes with $m=n$ balls}
\label{tab:previous_work}
\vspace{0.1in}
\end{minipage}

Another line of research on hash families focuses on studying functions with a constant evaluation time despite the expense of the number of random bits. For $O(\log n)$-wise independence, Siegel \cite{Siegel} showed how to implement it in constant time. For multiple-choice schemes, Woelfel \cite{Woelf} showed that the hash family of \cite{DW03}, which takes constant evaluation time and $n^{\Theta(1)}$ random bits, guarantees the same maximum loads as a perfectly random hash functions in the  multiple-choices schemes. 

P\v{a}tra\c{s}cu and Thorup \cite{PT12} introduced simple tabulation hashing, a function with constant evaluation time and $n^{\Theta(1)}$ random bits, that can replace the perfectly random hash functions in various applications. Recently, Dahlgaard et al. \cite{DKRT} prove that simple tabulation has the same \emph{expected} load balancing as a fully random hash function for 2-choice scheme, which was generalized to the \emph{Uniform-Greedy} and \emph{Always-Go-Left} schemes by Aamand, Knudsen, and Thorup \cite{AKT18}. However, Dahlgaard et al. \cite{DKRT} pointed out that for any constant $C$, with probability $n^{-\Theta(1)}$, the max-load becomes $C \cdot \log \log n$ in the \emph{Uniform-Greedy} with 2 choices.  

For the hash family designed by Celis et al. \cite{CRSW}, Meka et al.~\cite{MRRR} improved the evaluation time of \cite{CRSW} to $O((\log \log n)^2)$. 

\subsection{Discussion}
In this work, we provide a hash family with $O(\log n \log \log n)$ random bits that matches the maximum loads of a perfectly random hash function in multiple-choice schemes. A natural question is to reduce the number of random bits to $O(\log n)$. A starting point would be to improve the hash families in the $1$-choice scheme, where the best construction needs $O(\log n \log \log n)$ random bits from Celis et al.\cite{CRSW}.

For the $1$-choice scheme, the load of each bin is the summation of $m$ random indicator variables, which allows us to use the pseudorandom generators for concentration bounds \cite{SSS95,GKM15} and space-bounded computation \cite{Nisan92,NZ96,GKM15}. One interesting direction is to investigate the techniques of these two areas in the design of hash functions. At the same time, although Alon et al.\cite{ADMPT} proved lower bounds of $k$-wise independent functions in the 1-choice scheme, it is still interesting to explore natural algebraic constructions such as the quadratic characters of modulo $p$ for small-biased spaces in \cite{AGHP}.

Our work is an application of the technique --- milder restrictions \cite{CRSW,GMRTV} in the design of pseudorandom generators. Even though $k$-wise independence and small-biased spaces fool variants of classes, these two tools will not provide optimal pseudorandom generators for basic classes such as the 1-choice scheme \cite{ADMPT} or read-once CNFs \cite{DETT}. After Celis et al. \cite{CRSW} introduced milder restrictions, this technique has been successfully applied to construct almost optimal pseudorandom generators with $\log n \cdot \poly(\log \log n)$ random bits for several classes such as 1-choice scheme \cite{CRSW}, read-once CNFs \cite{GMRTV}, modulo $p$ functions and halfspaces \cite{GKM15}. It would be of great interest to investigate this technique to the design of pseudorandom generators for broader classes such as $AC^0$ circuits and space-bounded computation.

\subsection{Organization}
This paper is organized as follows. In Section~\ref{sec:pre}, we introduce some notation and tools. We define witness trees and revisit V\"{o}cking's argument \cite{Vocking} in Section~\ref{sec:witness_tree}. We show the construction of our hash family in Section~\ref{sec:hash_func} and sketch our derandomization in Section~\ref{sec:proof_sketch}. Next we prove Theorem~\ref{thm:intro_uniform} in Section~\ref{sec:uniform_scheme} and Theorem~\ref{thm:intro_left} in Section~\ref{asymmetric_tree}, which provide upper bounds on the maximum loads of the \emph{Uniform-Greedy} scheme and the \emph{Always-Go-Left} scheme separately. Finally, we prove Theorem~\ref{thm:intro_heavy_load} in Section~\ref{sec:heavy_load} which shows a bound of the heavy load case in the 1-choice scheme.


\section{Preliminaries}\label{sec:pre}
We use $U$ to denote the pool of balls, $m$ to denote the numbers of balls in $U$, and $n$ to denote the number of bins. We assume $m \ge n$ and $n$ is a power of $2$ in this work. We use $1_{E}$ to denote the indicator function of the event $E$ and $\mathbb{F}_p$ to denote the Galois field of size $p$ for a prime power $p$. 

\begin{definition}
Given a prime power p, a distribution $D$ on $\mathbb{F}_p^{n}$ is a $\delta$-biased space if for any non-trivial character function $\chi_\alpha$ in $\mathbb{F}_p^n$, $\underset{x \sim D}{\E}[\chi_\alpha(x)]\le \delta$.

A distribution $D$ on $\mathbb{F}_p^{n}$ is a $k$-wise $\delta$-biased space if for any non-trivial character function $\chi_\alpha$ in $\mathbb{F}_p^n$ of support size at most $k$, $\underset{x \sim D}{\E}[\chi_\alpha(x)]\le \delta$.
\end{definition}
The seminal works \cite{NN,AGHP} provide small-biased spaces with optimal seed length.
\begin{lemma}[\cite{NN,AGHP}]\label{lem:small_biased}
For any prime power $p$ and integer $n$, there exist explicit constructions of $\delta$-biased spaces on $\mathbb{F}_p^n$ with seed length $O(\log \frac{p n}{\delta})$ and explicit constructions of $k$-wise $\delta$-biased spaces with seed length $O(\log \frac{k p \log n}{\delta})$
\end{lemma}
Given two distributions $D_1$ and $D_2$ with the same support $\mathbb{F}_p^n$, we define the statistical distance to be $\|D_1-D_2\|_{1}=\sum_{x \in \mathbb{F}_p^n} |D_1(x)-D_2(x)|$. Vazirani \cite{Vazirani} proved that small-biased spaces are close to the uniform distribution.
\begin{lemma}[\cite{Vazirani}]\label{lem:small_biased_space}
A $\delta$-biased space on $\mathbb{F}_p^{n}$ is $\delta \cdot p^{n/2}$ close to the uniform distribution in  statistical distance.

Given a subset $S$ of size $k$ in $[n]$, a $k$-wise $\delta$-biased space on $\mathbb{F}_p^{n}$ is $\delta \cdot p^{k/2}$ close to the uniform distribution on $S$ in statistical distance.
\end{lemma}

Given a distribution $D$ on functions from $U$ to $[n]$, $D$ is $k$-wise independent if for any $k$ elements $x_1,\ldots,x_k$ in $U$, $D(x_1),\ldots,D(x_k)$ is a uniform distribution on $[n]^k$. For small-biased spaces, we choose $p=n$ and the space to be $\mathbb{F}_n^{|U|}$ in Lemma~\ref{lem:small_biased} and summarize the discussion above.
\begin{lemma}
Given $k$ and $n$, a $k$-wise $\delta$-biased space from $U$ to $[n]$ is $\delta \cdot n^{k/2}$ close to the uniform distribution from $U$ to $[n]$ on any $k$ balls, which needs $O(\log \frac{k n \log n}{\delta})$ random bits.
\end{lemma}
\begin{remark}
In this work, we always choose $\delta \le 1/n$ and $k=\poly(\log n)$ in the small biased spaces such that the seed length is $O(\log \frac{1}{\delta})$. At the same time, we only use $k$-wise small-biased spaces rather than small biased spaces to improve the evaluation time from $O(\log n)$ to $O(\log \log n)^4$.
\end{remark}
We state the Chernoff bound in $k$-wise independence by Schmidt et al. in \cite{SSS95}.
\begin{lemma}[Theorem 5 (I) (b) in \cite{SSS95}]\label{lem:k_wise_independence_chernoff}
If $X$ is the sum of $k$-wise independent random variables, each of which is confined to the interval $[0,1]$ with $\mu=\E[X]$, then for $\delta \le 1$ and $k \ge \delta^2 \mu \cdot e^{-1/3}$, $$\Pr[|X- \mu| \ge \delta \mu] \le e^{-\delta^2 \mu/3}.$$
\end{lemma}


\section{Witness Trees}\label{sec:witness_tree}
We first provide several notation and definitions in this section. Then we review the witness tree argument of V\"{o}cking \cite{Vocking} for the \emph{Uniform-Greedy} scheme.

\begin{definition}[\emph{Uniform-Greedy} with $d$ choices]
Let $h^{(1)},\ldots,h^{(d)}$ be $d$ hash functions from $U$ to $[n]$. The process inserts balls in any fixed order as follows:  for each ball $i$, the algorithm considers $d$ bins $\{h^{(1)}(i),\ldots,h^{(d)}(i)\}$ and puts the ball $i$ into the bin with the least load among $\{h^{(1)}(i),\ldots,h^{(d)}(i)\}$. When there are several bins with the least load, it picks an arbitrary one.
\end{definition}

For convenience, given $m$ balls $b_1,\cdots,b_m \in [m]$, we fix the order to be $b_1<b_2<\cdots<b_m$ in this work. We define the \emph{height} of a ball to be the height of it on the bin allocated in the above process.

Next we follow the notation of V\"{o}cking \cite{Vocking} to define witness trees and pruned witness trees. Given the balls and $d$ hash functions $h^{(1)},\ldots,h^{(d)}$ in the allocation process, we construct a symmetric witness tree for each ball in this process.

\begin{definition}[Symmetric witness trees]
Given a ball $b$ and a parameter $l$ less than the height of $b$, the symmetric witness tree $T$ with height $l$ for $b$ is a complete $d$-ary tree of height $l$. Every node $w$ in this tree corresponds to a ball $T(w) \in [n]$; and the root corresponds to the ball $b$. A ball $u$ in $T$ has a ball $v$ as its $i$th child iff when we allocate $u$ in the process, ball $v$ is the top ball in the bin $h^{(i)}\big(u\big)$. Hence $v<u$ and the bin $h^{(i)}(u)$ is in the subset $\left \{h^{(1)}(v),\ldots,h^{(d)}(v)\right\}$ of $[n]$ when $v$ is the $i$th child of $u$.
\end{definition}

Next we trim the repeated nodes in a witness trees such that there is no duplicate edge after the trimming. 

\begin{definition}[Pruned witness trees and collisions]\label{def:pruned_witness}
Given a witness tree $T$ where nodes $v_1,\ldots,v_j$ in $T$ correspond to the same ball, let $v_1$ be the node among them in the most bottom level of $T$. Consider the following process: first remove $v_2,\ldots,v_j$ and their subtrees; then, redirect the edges of $v_2,\ldots,v_j$ from their parents to $v_1$ and call these edges \emph{collisions}. Given a symmetric witness tree $T$, we call the new tree without repeated nodes after the above process as the pruned witness tree of $T$.
\end{definition}
We call different witness trees with the same structure but different balls a \emph{configuration}. For example, the configuration of symmetric witness trees with \emph{distinct} nodes is a full $d$-ary tree without any collision.

\begin{figure}[h]
\centering
\includegraphics{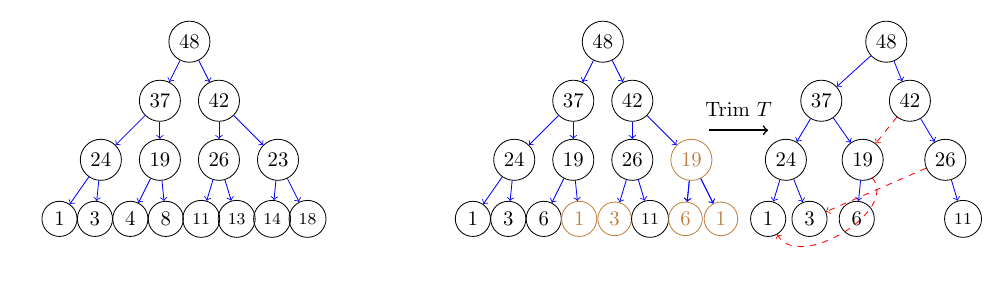}
\caption{A witness tree with distinct balls and a pruned witness tree with 3 collisions}
\label{fig1}
\end{figure}

Next we define the height and size of pruned witness trees.
\begin{definition}[Height of witness trees]
Given any witness tree $T$, let the height of $T$ be the length of the \emph{shortest} path from the root of $T$ to its leaves. Because $\text{height}(u)=\underset{v \in \text{children}(u)}{\min}\big\{\text{height}(v)\big\}+1$, the height of the pruned witness tree equals the height of the original witness tree. 

At the same time, let $|T|$ denote the number of vertices in $T$ for \emph{any} witness tree $T$ and $|C|$ denote the number of nodes in a configuration $C$.
\end{definition}
\begin{remark}
Given a ball $b$ of height $h$ and any $h'<h$, we always consider the pruned witness tree of $b$ with height $h'$ whose leaves have height at least $h-h'$.
\end{remark}

Finally we review the argument of V\"{o}cking \cite{Vocking} for $m=n$ balls. One difference between this proof and V\"{o}cking's original \cite{Vocking} proof is an alternate argument for the case of witness trees with many collisions.
\begin{lemma}[\cite{Vocking}]\label{lem:review_vocking}
For any constants $c \ge 2$ and $d$, with probability at least $1-n^{-c}$, the max-load of the \emph{Always-Go-Left} scheme with $d$ independent choices from perfectly random hash functions is $\frac{\log \log n}{\log d} + O(1)$.
\end{lemma}
\begin{proof}
We fix a parameter $l=\lceil \log_d \big( (2+2c)\log n \big)+ 3c + 5 \rceil$ for the height of witness trees. In this proof, we bound the probability that any symmetric witness tree of height $l$ with leaves of height at least 4 exists in perfectly random hash functions. From the definition of witness trees, this also bounds the probability of a ball with height $l+4$ in the $d$-choice \emph{Uniform-Greedy} scheme.

For symmetric witness trees of height $l$, it is sufficient to bound the probability that their pruned counterparts appear in perfectly random hash functions. We separate all pruned witness trees into two cases according to the number of edge collisions: pruned witness trees with at most $3c$ collisions and pruned witness trees with at least $3c$ collisions.

\paragraph{Pruned witness trees with at most $3c$ collisions.} Let us fix a configuration $C$ with at most $3c$ collisions and consider the probability any pruned witness trees with configuration $C$ appears in perfectly random hash functions. Because each node of this configuration $C$ corresponds a distinct ball, there are at most $n^{|C|}$ possible ways to instantiate balls into $C$.

Next, we fix one possible pruned witness tree $T$ and bound the probability of the appearance of $T$ in $h^{(1)},\ldots,h^{(d)}$. We consider the probability of two events: every edge $(u,v)$ in the tree $T$ appears during the allocation process; and every leaf of $T$ has height at least $4$. For the first event, an edge $(u,v)$ holds during the process only if the hash functions satisfy
\begin{equation}\label{eqn:hold}
h^{(i)}(u)\in \left \{h^{(1)}(v),\ldots,h^{(d)}(v)\right\}, \text{ which happens with probability at most } \frac{d}{n}.
\end{equation}

Secondly, the probability that a fixed leaf ball has height at least $4$ is at most $3^{-d}$. A leaf ball of height 4 indicates that each bin in his choices has height at least $3$. Because at most $n/3$ bins contain at least 3 balls at any moment, the probability that a random bin has height at least 3 is $\le 1/3$. Thus the probability that $d$ random bins have height $3$ is at most $3^{-d}$.

We apply a union bound on the probability that any witness tree with the configuration $C$ appears in perfectly random hash functions:
\begin{equation}\label{eq:vocking_at_most_collisions}
n^{|C|} \cdot \prod_{(u,v) \in C} \frac{d}{n} \cdot (3^{-d})^{\text{number of leaves}}
\end{equation}
We lower bound the number of edges in $C$ by $|C|-1$ because $C$ is connected. Next we lower bound the number of leaves. Because $C$ is a $d$-ary tree with at most $3c$ collisions, the number of leaves is at least $\frac{|C|-3c}{2}$. At the same time, $C$ is trimmed from the $d$-ary symmetric witness tree of height $l$. Thus $|C| \ge (1+d+\cdots+d^{l-3c})$. From all discussion above, we bound \eqref{eq:vocking_at_most_collisions} by
\[
n^{|C|} \cdot (\frac{d}{n})^{|C|-1} \cdot (3^{-d})^{\frac{|C|-3c}{2}} \le n \cdot (d^{2.5} \cdot 3^{-d})^{|C|/2.5} \le n \cdot (d^{2.5} \cdot 3^{-d})^{d^{l-3c}/2.5} \le n \cdot (0.8)^{10(2+2c) \log n} \le n^{-2c-1}.
\]

Finally, we apply a union bound on all possible configurations with at most $3c$ collisions: the number of configurations is at most $\sum_{i=0}^{3c} (d^{l+1})^{2 \cdot i} \le n$ such that the probability of any witness tree with height $l$ and at most $3c$ collisions existing is at most $n^{-c}$.

\paragraph{Pruned witness trees with at least $3c$ collisions.} We use the extra $3c$ collisions with equation \eqref{eqn:hold} instead of the number of leaves in this case.

Given any configuration $C$ with at least $3c$ collisions, we consider the first $3c$ collisions $e_1,\ldots,e_{3c}$ in the BFS of $C$. Let $C'$ be the induced subgraph of $C$ that only contains nodes in $e_1,\ldots,e_{3c}$ and their ancestors in $C$. At the same time, the size $|C'| \le 3c(2l+1)$ and the number of edges in $C'$ is $|C'|+3c-1$.
\begin{figure}[h]
\centering
\includegraphics{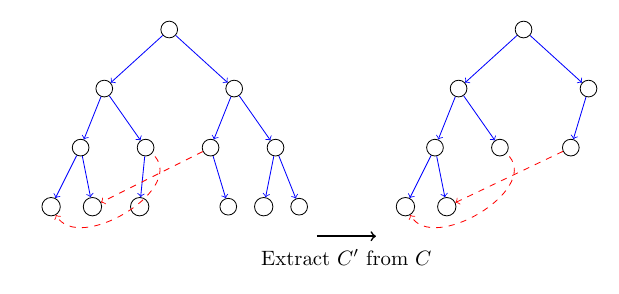}
\caption{An example of extracting $C'$ from $C$ given two collisions.}
\label{fig1}
\end{figure}

Because any pruned witness tree of $C$ exists only if its corresponding counterpart of $C'$ exists in perfectly random hash functions, it is suffice to bound the probability of the latter event. There are at most $n^{|C'|}$ instantiations of balls in $C'$. For each instantiation, we bound the probability that all edges survive by \eqref{eqn:hold}:
\[ (\frac{d}{n})^{\text{number of edges}}=(\frac{d}{n})^{|C'|+3c-1}. \] 

We bound the probability that any pruned witness tree of configuration $C'$ survives in the perfectly random hash function by
\[
(\frac{d}{n})^{|C'|+3c-1} \cdot n^{|C'|} \le (\frac{1}{n})^{3c-1} \cdot d^{(2l+2) \cdot 3c} \le n^{-2c}.
\]
Finally, we apply a union bound over all possible configurations $C'$: there are at most $(1+d+\cdots+d^l)^{2\cdot 3c} \le n$ configurations of $3c$ collisions.
\end{proof}
\begin{remark}
Because the sizes of all witness trees are bounded by $d^{l+1}=O(\log n)$, $O_{c,d}(\log n)$-wise independent hash functions could adopt the above argument to prove a max-load of $\log_d \log n + O(d+c)$.
\end{remark}


\section{Hash functions}\label{sec:hash_func}
We construct our hash family  and show its properties for the derandomization of witness trees argument in this section. We sketch the derandomization of Lemma~\ref{lem:review_vocking} of Vocking's argument in Section~\ref{sec:proof_sketch}.

Let $\circ$ denote the concatenation operation and $\oplus$ denote the bit-wise XOR operation.
\begin{construction}\label{def:hash}
Given $\delta_1>0, \delta_2>0$, and two integers $k,k_g$, let
\begin{enumerate}
\item $h_i: U \rightarrow [n^{2^{-i}}]$ denote a function generated by an $O(\log^2 n)$-wise $\delta_1$-biased space for each $i \in [k]$,
\item $h_{k+1}: U \rightarrow [n^{2^{-k}}]$ denote a function generated by an $O(\log^2 n)$-wise $\delta_2$-biased space such that $(h_1(x) \circ h_2(x) \circ \cdots \circ h_k(x) \circ h_{k+1}(x) )$ is a function by $U$ to $[n]$,
\item $g:U \rightarrow [n]$ denote a function from a $k_g$-wise independent family from $U$ to $[n]$.
\end{enumerate}
We define a random function $h:U \rightarrow [n]$ in our hash family $\mathcal{H}$ with parameters $\delta_1,\delta_2,k$ and $k_g$ to be:
\[ h(x) = \big(h_1(x) \circ h_2(x) \circ \cdots \circ h_k(x) \circ h_{k+1}(x) \big) \oplus g(x).\]
\end{construction}
Hence the seed length of our hash family is $O(k \log \frac{n \cdot \log^2 n \cdot \log |U|}{\delta_1} + \log \frac{n \cdot \log^2 n \cdot \log |U|}{\delta_2}+k_g \log n)$. We always choose $k \le \log \log n, k_g = O(\log \log n), \delta_1=1/\poly(n), $ and $\delta_2=(\log n)^{-O(\log n)}$ such that the seed length is $O(\log n \log \log n)$.

\begin{remark}
Our parameters of $h_1 \circ \cdots \circ h_{k+1}$ are stronger than the parameters in~\cite{CRSW}. While the last function $h_{k+1}$ of \cite{CRSW} is still a $\delta_1$-biased space, we use $\delta_2=(\delta_1)^{O(k)}$ in $h_{k+1}$ to provide almost $O(\log n)$-wise independence on $(\log n)^{O(\log n)}$ subsets of size $O(\log n)$ for our calculations.
\end{remark}

\paragraph{Properties of $h$.} We state the properties of $h$ that will be used in the derandomization. Because of the $k_g$-wise independence in $g$ and the $\oplus$ operation, we have the same property for $h$.
\begin{property}\label{clm:loglogn_wise}
$h$ is $k_g$-wise independent.
\end{property}

Then we fix $g$ and discuss $h_1 \circ \cdots \circ h_k \circ h_{k+1}$. For each $i \in [k]$, it is natural to think of $h_1 \circ \cdots \circ h_i$ as a function from $U$ to $[n^{1-\frac{1}{2^i}}]$, i.e., a hash function maps all balls into $n^{1-\frac{1}{2^i}}$ bins. Celis et al. \cite{CRSW} showed that for every $i \in [k]$, the number of balls in every bin of $h_1 \circ \cdots \circ h_i$ is close to its expectation $n^{\frac{1}{2^i}} \cdot \frac{m}{n}$ in $\frac{1}{\poly(n)}$-biased spaces. 
\begin{lemma}[\cite{CRSW}]\label{lem:uniform_distributed}
Given $k=\log_2 (\log n/3\log \log n)$ and $\beta=(\log n)^{-0.2}$, for any constant $c>0$, there exists $\delta_1=1/\poly(n)$ such that given $m=O(n)$ balls, with probability at least $1-n^{-c}$, for all $i \in [k]$, every bin in $[n^{1-\frac{1}{2^i}}]$ contains at most $(1+\beta)^i n^{\frac{1}{2^i}} \cdot \frac{m}{n}$ balls under $h_1 \circ \cdots \circ h_i$.
\end{lemma}

For completeness, we provide a proof of Lemma~\ref{lem:uniform_distributed} in Appendix~\ref{app:proof}. In this work, we use the following version that after fixing $g$ in the Construction~\ref{def:hash}, $h_1 \circ h_2 \circ \cdots \circ h_k$ still allocates the balls evenly.
\begin{corollary}\label{cor:last_hash}
For any constant $c>0$, there exists $\delta_1=1/\poly(n)$ such that given $m=O(n)$ balls and any function $g_0:U \rightarrow [n/\log^3 n]$, with probability at least $1-n^{-c}$ over $h_1,\ldots,h_k$, for any bin $j \in [n^{1-\frac{1}{2^k}}]=[n/\log^3 n]$, it contains at most $1.01 \cdot \log^3 n \cdot \frac{m}{n}$ balls in the hash function $\big(h_1(x) \circ \cdots \circ h_k(x) \big) \oplus g_0(x)$.
\end{corollary}

Next we discuss the last function $h_{k+1}$ generated from a $\delta_2$-biased space on $[\log^3 n]^U$. For a subset $S \subseteq U$, let $h(S)$ denote the distribution of a random function $h$ on $S$ and $U_{[\log^3 n]}(S)$ denote the uniform distribution over all maps from $S \rightarrow [\log^3 n]$. From Lemma~\ref{lem:small_biased_space} and the union bound, we have the following claim.
\begin{claim}\label{clm:SD_subsets}
Given $\delta_2=(\log n)^{-C \log n}$, for a fixed subset $S$ of size $\frac{C}{3} \cdot \log n$, $h_{k+1}(S)$ is $(\log n)^{- \frac{C}{2} \cdot \log n}$-close to the uniform distribution on $S$, i.e., $\|h_{k+1}(S)-U_{[\log^3 n]}(S)\|_{1} \le (\log n)^{- \frac{C}{2} \cdot \log n}$.

Then for $m=(\log n)^{\frac{C}{3} \cdot \log n}$ subsets $S_1,\ldots,S_m$ of size $\frac{C}{3} \cdot \log n$, we have $$\sum_{i \in [m]} \|h_{k+1}(S_i)-U_{[\log^3 n]}(S_i)\|_{1} \le m \cdot (\log n)^{- \frac{C}{2} \cdot \log n} \le (\log n)^{- \frac{C}{6} \cdot \log n}.$$
\end{claim}
In another word, $h_{k+1}$ is close to the uniform distribution on \emph{$(\log n)^{\frac{C}{3} \log n}$ subsets of size $\frac{C}{3} \log n$}. However, $h_{k+1}$ (or $h$) is not close to $\log n$-wise independence on $n$ balls.


\begin{remark}[Evaluation time]
Our hash function has an evaluation time $O((\log \log n)^4)$ in the RAM model. Because we use $(\log n)^{-O(\log \log n)}$-biased spaces in $h_{k+1}$, we lose a factor of $O(\log \log n)^2$ compared to the hash family of \cite{CRSW}. The reason is as follows.

$g$ can be evaluated by a degree $O(\log \log n)$ polynomial in the Galois field of size $\poly(n)$, which takes $O(\log \log n)$ time. The first $k$ hash functions $h_1,\ldots,h_k$ use $1/\poly(n)$-biased spaces, which have total evaluation time $O(k \cdot \log \log n)=O(\log \log n)^2$ in the RAM model from \cite{MRRR}.

The last function $h_{k+1}$ in the RAM model is a $O(\log n)$-wise $n^{-O(\log \log n)}$-biased space from $U$ to $[\log^3 n]$, which needs $O(\log \log n)$ words in the RAM model. Thus the evaluation time becomes $O(\log \log n)$ times the cost of a quadratic operation in the Galois field of size $n^{O(\log \log n)}$, which is $O((\log \log n)^4)$.
\end{remark}

\subsection{Proof Overview}\label{sec:proof_sketch}
We sketch the derandomization of Lemma~\ref{lem:review_vocking} in this section. Similar to the proof of Lemma~\ref{lem:review_vocking}, we bound the probability that any pruned witness tree of height $l=\log_d \log n + O(1)$ exists in $h^{(1)},\ldots,h^{(d)}$, where each $h^{(i)}=\big(h^{(i)}_1(x) \circ \cdots \circ h^{(i)}_{k+1}(x) \big) \oplus g^{(i)}(x)$. We use the property of $h_1 \circ \cdots \circ h_{k+1}$ to derandomize the case of pruned witness trees with at most $3c$ collisions and the property of $g$ to derandomize the other case. 

\paragraph{Pruned witness trees with at most $3c$ collisions.} We show how to derandomize the union bound \eqref{eq:vocking_at_most_collisions} for a fixed configuration $C$ with at most $3c$ collisions. There are two probabilities in \eqref{eq:vocking_at_most_collisions}: the second term $\prod_{(u,v) \in C} \frac{d}{n}$ over all edges in $C$ and the last term $3^{-d \cdot \text{number of leaves}}$ over all leaves. We focus on the first term $\prod_{(u,v) \in C} \frac{d}{n}$ in this discussion, because it contributes a smaller probability. Since $|C| \in [d^{l-3c}, d^{l+1}] = \Theta(\log n)$, it needs $O(\log n)$-wise independence over $[n]$ bins for every possible witness trees in \eqref{eq:vocking_at_most_collisions}, which is impossible to support with $o(\log^2 n)$ bits \cite{Stinson94b}.

We omit $\{g^{(1)},\ldots,g^{(d)}\}$ and focus on the other part $\big\{h^{(i)}_1 \circ \cdots \circ h^{(i)}_k \circ h^{(i)}_{k+1}\big| i \in [d]\big\}$ in this case. Our strategy is to first fix the prefixes in the $d$ hash functions, $\big\{h^{(i)}_1 \circ \cdots \circ h^{(i)}_k\big| i \in [d]\big\}$, then recalculate \eqref{eq:vocking_at_most_collisions} using the suffixes $h^{(1)}_{k+1},\ldots,h^{(d)}_{k+1}$. Let $T$ be a possible witness tree in the configuration $C$. To satisfy \eqref{eqn:hold} for an edge $(u,v)$ in $T$, the prefixes of $h^{(1)}(v),\ldots,h^{(d)}(v)$ and $h^{(i)}(u)$ must satisfy
\begin{equation}\label{eq:precondition}
h^{(i)}_1(u)  \circ \cdots \circ h^{(i)}_k (u) \in \left\{h^{(1)}_1(v)  \circ \cdots \circ h^{(1)}_k(v),\ldots, h^{(d)}_1(v) \circ \cdots \circ h^{(d)}_k(v)\right\}.
\end{equation}

After fixing the prefixes, let $\mathcal{F}_T$ denote the subset of possible witness trees in the configuration $C$ that satisfy the prefix condition~\eqref{eq:precondition} for every edge. Because each bin of $[n/\log^3 n]$ receives at most $1.01 \log^3 n$ balls from every prefix function $h_1^{(j)} \circ h_2^{(j)} \circ \cdots \circ h_k^{(j)}$ by Corollary~\ref{cor:last_hash}, we could bound
$$|\mathcal{F}_T| \le n (d \cdot 1.01  \log^3 n)^{|C|-1}=n \cdot (1.01 d)^{|C|} \cdot (\log^3 n)^{|C|-1}=(\log n)^{O(\log n)}$$
instead of $n^{|C|}$ in the original argument.

Now we consider all possible witness trees in $\mathcal{F}_T$ under the suffixes $h^{(1)}_{k+1},\ldots,h^{(d)}_{k+1}$. We could treat $h^{(1)}_{k+1},\ldots,h^{(d)}_{k+1}$ as $O(\log n)$-wise independent functions for all possible witness trees in $\mathcal{F}_T$ from Claim~\ref{clm:SD_subsets}, because $|C|=O(\log n)$ and $|\mathcal{F}_T|=(\log n)^{O(\log n)}$. In the next step, we use $O(\log n)$-wise independence to rewrite \eqref{eq:vocking_at_most_collisions} and finish the proof of this case.




\paragraph{Pruned witness trees with at least $3c$ collisions.} In our alternate argument of this case in Lemma~\ref{lem:review_vocking}, the subconfiguration $C'$ of $C$ has at most $3c \cdot (2l+1)$ nodes and $3c \cdot (2l+1)+3c$ edges. Since $l=\log_d \log n + O(1)$, the number of edges in $C'$ is $O(\log \log n)$. By choosing $k_g=\Theta(\log \log n)$ with a sufficiently large constant, $h$ with $k_g$-wise independence supports the argument in Lemma~\ref{lem:review_vocking}.

\section{The \emph{Uniform Greedy} scheme}\label{sec:uniform_scheme}
We prove our main result for the \emph{Uniform-Greedy} scheme --- Theorem~\ref{thm:intro_uniform} in this section.

\begin{theorem}\label{thm:sym}
For any $m=O(n)$, any constant $c \ge 2$, and integer $d$, there exists a hash family $\mathcal{H}$ from Construction~\ref{def:hash} with $O(\log n \log \log n)$ random bits that guarantees the max-load of the \emph{Uniform Greedy} scheme with $d$ independent choices from $\mathcal{H}$ is $\log_d \log n + O\big(c + \frac{m}{n}\big)$ with probability at least $1-n^{-c}$ for any $m$ balls in $U$.
\end{theorem}

\begin{proof}
We specify the parameters of $\mathcal{H}$ as follows: $k_g=10c(\log_d \log m + \log_d (2+2c) + 5 + 3c ), k=\log_2 \frac{\log n}{3 \log \log n}$, $\delta_2=\log n^{-C\log n}$ for a large constant $C$, and $\delta_1=1/\poly(n)$ such that Corollary \ref{cor:last_hash} holds with probability at least $1-n^{-c-1}$. Let $h^{(1)}, \ldots, h^{(d)}$ denote the $d$ independent hash functions from $\mathcal{H}$ with the above parameters, where each \[ h^{(j)}(x) = \big(h^{(j)}_1(x) \circ h^{(j)}_2(x) \circ \cdots \circ h^{(j)}_k(x) \circ h^{(j)}_{k+1}(x) \big) \oplus g^{(j)}(x).\] We use the notation $g$ to denote $\{g^{(1)},g^{(2)},\ldots,g^{(d)}\}$ in the $d$ choices and $h_{i}$ to denote the group of hash functions $\{h^{(1)}_{i},\ldots, h^{(d)}_i\}$ in this proof. 

We bound the probability that any symmetric witness tree of height $l=\lceil \log_d \log m + \log_d (2+2c) + 5 + 3c \rceil$ with leaves of height at least $b=\ConstHeight$ exists in $h^{(1)}, \ldots, h^{(d)}$. Similar to the proof of Lemma~\ref{lem:review_vocking}, we bound the probability of pruned witness trees of height $l$ in $h^{(1)}, \ldots, h^{(d)}$. We separate all pruned witness trees into two cases according to the number of edge collisions: pruned witness trees with at most $3c$ collisions and pruned witness trees with at least $3c$ collisions.

\paragraph{Pruned witness trees with at least $3c$ collisions.} We start with a configuration $C$ of pruned witness trees with height $l$ and at least $3c$ collisions. Let $e_1,\ldots,e_{3c}$  be the first $3c$ collisions in the BFS of $C$. Let $C'$ be the induced subgraph of $C$ that only contains nodes in these edges $e_1,\ldots,e_{3c}$ and their ancestors in $C$. Therefore any pruned witness tree $T$ of configuration $C$ exists in $h^{(1)},\ldots,h^{(d)}$ only if the corresponding counterpart $T'$ of $T$ with configuration $C'$ exists in $h^{(1)},\ldots,h^{(d)}$. The existence of $T'$ in $h^{(1)},\ldots,h^{(d)}$ indicates that for every edge $(u,v)$ in $T'$, $h^{(1)},\ldots,h^{(d)}$ satsify
\begin{equation}\label{eq:k_g_wise}
h^{(i)}\big(T(u)\big) \in \left\{h^{(1)}\big(T(v)\big),\ldots,h^{(d)}\big(T(v)\big)\right\} \text{ when $v$ is the $i$th child of $u$}.
\end{equation}
Notice that the number of edges in $C'$ and $T'$ is at most $3c \cdot 2l+3c=2l(3c+1) \le k_g/2$.

Because $h^{(1)},\ldots,h^{(d)}$ are $k_g$-wise independent, We bound the probability that all edges of $T'$ satisfy \eqref{eq:k_g_wise} in $h^{(1)},\ldots,h^{(d)}$ by
\[
\prod_{(u,v) \in T'}(\frac{d}{n})=(\frac{d}{n})^{|C'|+3c-1}.
\]
Now we apply a union bound over all choices of balls in $C'$. There are at most $m^{|C'|}$ choices of balls in the nodes of $C'$. Therefore we bound the probability that any witness with at least $3c$ collisions survives in $k_g$-wise independent functions by
\[
(\frac{d}{n})^{|C'|+3c-1} \cdot m^{|C'|} \le (\frac{d}{n})^{3c-1} \cdot (\frac{m}{n} \cdot d)^{|C'|} \le (\frac{d}{n})^{3c-1} \cdot (\frac{m}{n} \cdot d)^{3c \cdot (2l+1)} \le n^{-2c}.
\]

Next we apply a union bound over all configurations $C'$. Because there are at most $(1+d+\cdots+d^l)^{2\cdot 3c} \le n$ configurations of $3c$ collisions, with probability at least $1-n^{-c}$, there is no pruned witness trees with at least $3c$ collision and height $l$ exists in $h^{(1)},\ldots,h^{(d)}$.

\paragraph{Pruned witness trees with at most $3c$ collisions.}
We fix a configuration $C$ of pruned witness trees with height $l$ and less than $3c$ collisions. Next we bound the probability that any pruned witness trees in this configuration $C$ with leaves of height at least $b$ exists in $h^{(1)},\ldots,h^{(d)}$.

We extensively use the fact that after fixing $g$ and $h_1 \circ \cdots \circ h_k$, at most $d (1.01 \log^3 n \cdot \frac{m}{n})$ elements in $h^{(1)},\ldots,h^{(d)}$ are mapped to any bin of $[n/\log^3 n]$ from Corollary~\ref{cor:last_hash}. Another property is the number of leaves in $C$: because there are at most $3c$ collisions in $C$, $C$ has at least $d^{l-3c} \in [d^5(2+2c) \log m, d^6(2+2c) \log m]$ leaves. On the other hand, the number of leaves is at least $\frac{|C|-3c}{2}$.

For a pruned witness tree $T$ with configuration $C$, $T$ exists in $h^{(1)},\ldots,h^{(d)}$ only if
\begin{equation}\label{eq:t_wise}
\forall (u,v) \in C, h^{(i)}\big(T(u)\big) \in \left\{h^{(1)}\big(T(v)\big),\ldots,h^{(d)}\big(T(v)\big)\right\} \text{ when $v$ is the $i$th child of $u$}.
\end{equation}
We restate the above condition  on the prefixes and suffixes of $h^{(1)},\ldots,h^{(d)}$ separately. Let $g_{p}(x)$ denote the first $\log n - 3 \log \log n$ bits of $g(x)$ and $g_{s}(x)$ denote the last $3 \log \log n$ bits of $g(x)$, which matches $h_1(x) \circ \cdots \circ h_k(x)$ and $h_{k+1}(x)$ separately. Since $h^{(i)}(x)=\big(h^{(i)}_1(x) \circ \cdots \circ h^{(i)}_{k+1}(x) \big) \oplus g^{(i)}(x)$, property \eqref{eq:t_wise} indicates that the prefixes of the balls $b_u=T(u)$ and $b_v=T(v)$ satisfy
\begin{equation}\label{eq:prefix}
\big(h^{(i)}_1(b_u)   \circ \cdots \circ h^{(i)}_k (b_u) \big) \oplus g^{(i)}_{p}(b_u) \in \left\{ \big(h^{(1)}_1(b_v)  \circ \cdots \circ h^{(1)}_k(b_v) \big) \oplus g^{(1)}_{p}(b_v),\ldots, \big(h^{(d)}_1(b_v) \circ \cdots \circ h^{(d)}_k(b_v) \big) \oplus g^{(d)}_{p}(b_v)\right\}.
\end{equation}
and their suffixes satisfy
\begin{equation}\label{eq:suffix}
h^{(i)}_{k+1}(b_u) \oplus g^{(i)}_s(b_u) \in \left\{h^{(1)}_{k+1}(b_v) \oplus g^{(1)}_s(b_v),\ldots, h^{(d)}_{k+1}(b_v) \oplus g^{(d)}_s(b_v)\right\}.
\end{equation}

Let $\mathcal{F}_T$ be the subset of witness trees in the configuration $C$ whose edges satisfy the condition \eqref{eq:prefix} in preffixes $h_{(1)},\ldots,h_{(k)}$, i.e., $\mathcal{F}_T=\{T|\text{configuration}(T)=C \text{ and } (u,v) \text{ satisfies \eqref{eq:prefix} } \forall (u,v) \in T\}$. We show that
\[ |\mathcal{F}_T| \le m \cdot (d \cdot 1.01 \log^3 n \cdot \frac{m}{n})^{|C|-1}.\]
The reason is as follows. There are $m$ choices of balls for the root $u$ in $C$. For the $i$th child $v$ of the root $u$, we have to satisfy the condition \eqref{eq:prefix} for $(u,v)$. For a  fixed bin $\big(h^{(i)}_1(b_u)   \circ \cdots \circ h^{(i)}_k (b_u) \big) \oplus g^{(i)}_{p}(b_u)$, there are at most $1.01 \cdot \log^3 n \cdot \frac{m}{n}$ elements from each hash function $h^{(j)}$ mapped to this bin from Corollary~\ref{cor:last_hash}. Hence there are at most $d \cdot 1.01 \log^3 n \cdot \frac{m}{n}$ choices for each child of $u$. Then we repeat this arguments for all non-leaf nodes in $C$.

Next we consider the suffixes $h^{(1)}_{k+1},\ldots,h^{(d)}_{k+1}$. We first calculate the probability that any possible witness tree in $\mathcal{F}_T$ survives in $h^{(1)}_{k+1},\ldots,h^{(d)}_{k+1}$ from $t$-wise independence for $t=5b \cdot d^{l+2}=O(\log n)$. After fixing $g_s$,  for a possible witness tree $T$ in $\mathcal{F}_T$, $h^{(1)}_{k+1}, \ldots,h^{(d)}_{k+1}$ satisfy \eqref{eq:suffix} for every edge $(u,v) \in C$ with probability $\frac{d}{\log^3 n}$ in $t/2$-wise independent distributions because the number of edges in $C$ is less than $t/2$.

For each leaf $v$ in $T$, we bound the probability that its height is at least $b=\ConstHeight$ by $2^{-3d^2} \cdot (\frac{n}{m})^{2d}$ in $(b \cdot d+1)$-wise independence. Given a choice $i \in [d]$ of leaf $v$, we fix the bin to be $h^{(i)}(v)$. Then we bound the probability that there are at least $b-1$ balls $w_1,\ldots,w_{b-1}$ in this bin excluding all balls in the tree by
\begin{align*}
 \sum_{w_1: w_1<v, w_1 \notin T} \sum_{w_2: w_1<w_2<v,w_2 \notin T} \cdots \sum_{w_{b-1}:w_{b-2}<w_{b-1}<v,w_{b-1} \notin T} \Pr[h^{(i)}(v)=h^{(j_1)}(w_1)=\cdots=h^{(j_{b-1})}(w_{b-1})]\\ \le \frac{{1.01 d \cdot \log^3 n \cdot \frac{m}{n} \choose b-1 }}{(\log^3 n)^{b-1}} \le \frac{(1.01 d \cdot \frac{m}{n})^{b-1}}{(b-1)!}\le (\frac{3}{4})^{b-1}.
\end{align*}
For all $d$ choices of this leaf $v$, this probability is at most $(\frac{3}{4})^{(b-1) \cdot d}\le 2^{-3d^2} \cdot (\frac{n}{m})^{2d}$.

Because $w_1,\ldots,w_b$ are not in the tree $T$ for every leaf, they are disjoint and independent with the events of \eqref{eq:suffix} in $T$, which are over all edges in the tree. Hence we could multiply these two probability together in $t$-wise independence given $t/2 \ge (b \cdot d +1) \cdot \text{number of leaves}$. Then we apply a union bound over all possible pruned witness trees in $\mathcal{F}_T$ to bound the probability (in the $t$-wise independence) that there is one witness tree of height $l$ whose leaves have height at least $\ConstHeight$ by
\begin{align*}
|\mathcal{F}_T| \cdot (\frac{d}{\log^3 n})^{|C|-1} \cdot \big((\frac{3}{4})^{b \cdot d}\big)^{\text{number of leaves}} \le & m \left(1.01d \cdot \log^3 n \cdot \frac{m}{n} \cdot \frac{d}{\log^3 n}\right)^{|C|} \cdot \left(2^{-3d^2} \cdot (\frac{n}{m})^{2d}\right)^{\frac{|C|-3c}{2}} \\
\le & m \cdot \left(2d^2 \cdot \frac{m}{n} \right)^{|C|} \cdot \left(2^{-3d^2} \cdot (\frac{n}{m})^{2d} \right)^{|C|/3} \le m \cdot 2^{-|C|/3} \le n^{-c-1}.
\end{align*}

Finally we replace the $t$-wise independence by a $\delta_2$-biased space for $\delta_2 = n^{-c-1} \cdot (\log^3 n)^{-t}/|\mathcal{F}_T| = (\log n)^{-O(\log n)}$. We apply Claim~\ref{clm:SD_subsets} to all possible pruned witness tress in $\mathcal{F}_T$: in $\delta_2$-biased spaces, the probability of the existence of any height-$l$ witness tree with leaves of height at least $b=\ConstHeight$ is at most
\[ n^{-c-1} + |\mathcal{F}_T| \cdot \delta_2 \cdot (\log^3 n)^{t} \le 2 n^{-c-1}.\] Then we apply a union bound on all possible configurations with at most $3c$ collisions:
\[ (d^{l+1})^{|3c|} \cdot 2n^{-c-1} \le 0.5 n^{-c}. \]

From all discussion above, with probability at least $1-n^{-c}$, there is no ball of height more than $l+b=\log_d \log n + O(1)$.
\end{proof}



\section{The \emph{Always-Go-Left} Scehme}\label{asymmetric_tree}
We show that the hash family in Section~\ref{sec:hash_func} with proper parameters also achieves a max-load of $\frac{\log\log n}{d \log \phi_d } + O(1)$ in the \emph{Always-Go-Left} scheme \cite{Vocking} with $d$ choices, where $\phi_d>1$ is the constant satisfying $\phi_d^d=1 + \phi_d + \cdots + \phi_d^{d-1}$. We define the \emph{Always-Go-Left} scheme~\cite{Vocking} as follows:

\begin{definition}[\emph{Always-Go-Left} with $d$ choices]
Our algorithm partitions the bins into $d$ groups $G_1,\ldots,G_d$ of the same size $n/d$. Let $h^{(1)},\ldots,h^{(d)}$ be $d$ functions from $U$ to $G_1,\ldots,G_d$ separately. For each ball $b$, the algorithm considers $d$ bins $\{h^{(1)}(b) \in G_1, \ldots, h^{(d)}(b) \in G_d\}$ and chooses the bin with the least number of balls. If there are several bins with the least number of balls, our algorithm always chooses the bin with the smallest group number.
\end{definition}

We define asymmetric witness trees for the \emph{Always-Go-Left} mechanism such that a ball of height $l+C$ in the \emph{Always-Go-Left} scheme indicates that there is an asymmetric witness tree of height $l$ whose leaves have height at least $C$. For an asymmetric witness tree $T$, the height of $T$ is still the \emph{shortest} distance from the root to its leaves.
\begin{definition}[Asymmetric Witness tree]
The asymmetric witness tree $T$ of height $l$ in group $G_i$ is a $d$-ary tree. The root has $d$ children where the subtree of the $j$th child is an asymmetric witness tree in group $G_j$ of height $(l-1_{j\ge i})$.

Given $d$ functions $h^{(1)},\ldots,h^{(d)}$ from $U$ to $G_1,\ldots,G_d$ separately, a ball $b$ with height more than $l+C$ in a bin of group $G_i$ indicates an asymmetric witness tree $T$ of height $l$ in $G_i$ whose leaves have height at least $C$. Each node of $T$ corresponds to a ball, and the root of $T$ corresponds to the ball $b$. A ball $u$ in $T$ has a ball $v$ as its $j$th child iff when we insert the ball $u$ in the \emph{Always-Go-Left} mechanism, $v$ is the top ball in the bin $h^{(j)}(u)$. Hence $v<u$ and $h^{(j)}(u)=h^{(j)}(v)$ when the $j$th child of $u$ is $v$.
\end{definition}
For an asymmetric witness tree $T$ of height $l$ in group $G_i$, We use the height $l$ and the group index $i \in [d]$ to determine its size. Let $f(l,i)$ be the size of a \emph{full} asymmetric witness tree of height $l$ in group $G_i$. From the definition, we have $f(0,i)=1$ and
\[ 
f(l,i) = \sum_{j=1}^{i-1} f(l,j) + \sum_{j=i}^d f(l-1,j).
\]
Let $g\big((l-1) \cdot d+i\big)=f(l,i)$ such that 
\[
g(n)=g(n-1) + g(n-2) + \cdots + g(n-d).
\] 
We know there exist $c_0>0$, $c_1=O(1)$, and $\phi_d>1$ satisfying
\[
\phi_d^d = 1 + \phi_d + \cdots + \phi_d^{d-1} \text{ such that } g(n) \in [c_0 \cdot \phi_d^n,c_1 \cdot \phi_d^n].
\]
Hence $$f(l,i) = g\big((l-1) \cdot d+i\big) \in [c_0 \cdot \phi_d^{(l-1)d+i},c_1 \cdot \phi_d^{(l-1)d+i}] .$$ Similar to the pruned witness tree of a symmetric witness tree, we  use the same process in Definition~\ref{def:pruned_witness} to obtain the pruned asymmetric witness tree of an asymmetric witness tree.

V\"ocking in \cite{Vocking} showed that in a perfectly random hash function, the maximum load is $\frac{\log \log n}{d \log \phi_d} + O(1)$ with high probability given any $n$ balls. We outline V\"ocking's argument for \emph{distinct balls} here: let $b$ be a ball of height $l+4$ for $l=\frac{\log \log n + \log (1+c) }{d \log \phi_d}+1$. Without loss of generality, we assume that $b$ is in the first group $G_1$. By the definition of the asymmetric witness tree, there exists a tree $T$ in $G_1$ with root $b$ and height $l$ whose leaves have height at least $4$. For each ball $u$ and its $i$th ball $v$, the hash function $h^{(i)}$ satisfies $h^{(i)}(u)=h^{(i)}(v)$. Similar to \eqref{eq:vocking_at_most_collisions}, we apply a union bound on all possible witness trees of height $l$ in this configuration to bound the probability by
\[ n^{f(l,1)} \cdot (\frac{d}{n})^{f(l,1)-1} \cdot (\frac{1}{3^d})^{\text{number of leaves in }f(l,1)},\]
which is less than $n^{-c}$ given $f(l,1) = \Theta(\phi_d^{(l-1)d+1}) = \Theta\big((1+c)\log n\big)$.

We prove our derandomization of V\"ocking's argument here.
\begin{theorem}\label{thm:goleft}
For any $m=O(n)$, any constants $c>1$ and $d \ge 2$, there exist a constant $\phi_d \in (1.61,2)$ and a hash family $\mathcal{H}$ in Construction~\ref{def:hash} with $O(\log n \log \log n)$ random bits such that for any $m$ balls in $U$, with probability at least $1-n^{-c}$, the max-load of the \emph{Always-Go-Left} mechanism with $d$ independent choices from $\mathcal{H}$ is $\frac{\log \log n }{d \log \phi_d} + O(c+\frac{m}{n})$.
\end{theorem}
\begin{proof}
Let $l$ be the smallest integer such that $c_0 \phi_d^{l d} \ge 10(2+2c) \log m$ and $b=\ConstHeight$. We bound the probability of a witness tree of height $l+3c+1$ whose leaves have height more than $b$ in $h^{(1)},\ldots,h^{(d)}$ during the \emph{Always-Go-Left} scheme. Notice that there is a ball of height $l+b+3c+1$ in any bin of $G_2,G_3,\ldots,G_d$ indicates that there is a ball of the same height in $G_1$.

We choose the parameters of $\mathcal{H}$ as follows: $k_g=20c \cdot d \cdot (l+b+1+3c)=O(\log \log n), k=\log_2 (\log n/3\log \log n)$, $\delta_1=1/\poly(n)$ such that Corollary~\ref{cor:last_hash} happens with probability at most $n^{-c-1}$, and the bias $\delta_2=\log n^{-O(\log n)}$ of $h_{k+1}$ later. We set $h_{k+1}$ to be a hash function from $U$ to $[\log^3 n/d]$ and $g$ to be a function from $U$ to $[n/d]$ such that 
$$
h^{(j)}=\big( h^{(j)}_1 \circ h^{(j)}_1 \circ \cdots \circ h^{(j)}_k \circ h^{(j)}_{k+1} \big) \oplus g^{(j)}
$$ 
is a map from $U$ to $G_j$ of $[n/d]$ bins for each $j \in d$.

We use $h^{(1)}, \ldots, h^{(d)}$ to denote $d$ independent hash functions from $\mathcal{H}$ with the above parameters. We use the notation of $h_{i}$ to denote the group of hash functions $\{h^{(1)}_{i},\ldots, h^{(d)}_i\}$ in this proof. We assume Corollary~\ref{cor:last_hash} and follow the same argument in the proof of Theorem \ref{thm:sym}. We bound the probability of witness trees from 2 cases depending on the number of collisions.

\paragraph{Pruned witness trees with at least $3c$ collisions:}  Given a configuration $C$ with at least $3c$ collisions, we consider the first $3c$ collisions $e_1,\ldots,e_{3c}$ in the BFS of $C$. Let $C'$ be the induced subgraph of $C$ that only contains all vertices in $e_1,\ldots,e_{3c}$ and their ancestors in $C$. Therefore $C$ survives under $h^{(1)},\ldots,h^{(d)}$ only if $C'$ survives under $h^{(1)},\ldots,h^{(d)}$.

Observe that $|C'| \le 3c \cdot 2 \cdot \big(d \cdot \text{height}(T)\big)$. There are at most $m^{|T'|}$ possible instantiations of balls in $C'$. For each instantiation $T$ of $C'$, because $k_g \ge 2 \cdot \text{number of edges}=2( |C'|+3c-1)$, we bound the probability that any instantiation of $C'$ survives in $h$ by
\[ 
m^{|C'|} \cdot (\frac{d}{n})^{\text{number of edges}}=m^{|C'|} \cdot (\frac{d}{n})^{|C'|+3c-1} \le (d m/n)^{|C'|} \cdot (\frac{d}{n})^{3c-1} \le n^{-2c}. 
\] 
At the same time, there are at most $(|T|^2)^{3c}=\poly(\log n)$ configurations of $C'$. Hence we bound the probability of any witness with at least $3c$ collisions surviving by $n^{-c}$.

\paragraph{Pruned witness tree with less than $3c$ collisions:}
We fix a configuration $C$ of witness tree in group $G_1$ with height $l+1+3c$ and less than $3c$ collisions.  Thus $|C| \in [f(l+1,1),f(l+1+3c,1)]$.

Let $\mathcal{F}_T$ be the subset of possible asymmetric witness tree with configuration $C$ after fixing the prefixes $h_1, h_2, \ldots, h_k$. For any $T \in \mathcal{F}_T$, each edge $(u,v)$ has to satisfy $h^{(i)}\big(T(u)\big)=h^{(i)}\big(T(v)\big)$ in the \emph{Always-Go-Left} scheme when $v$ is the $i$th child of $u$. This indicates their prefixes are equal:
\[
h^{(i)}_1\big(T(u)\big) \circ \cdots \circ h^{(i)}_k\big(T(u)\big)=h^{(i)}_1\big(T(v)\big) \circ \cdots \circ h^{(i)}_k\big(T(v)\big).
\]
From the same argument in the proof of Theorem~\ref{thm:sym}, we bound \[|\mathcal{F}_T| \le m \cdot (1.01 \log^3 n \cdot \frac{m}{n})^{|C|-1}\] under $h_1, h_2,\ldots, h_k$ from Corollary~\ref{cor:last_hash}.

We first consider $h_{k+1}$ as a $t$-wise independent distribution from $U$ to $[\log^3 n/d]$ for $t=5bd \cdot f(l+3c+1,1)=O(\log m)$ then move to $\delta_2$-biased spaces. For each asymmetric witness tree, every edge $(u,v)$ maps to the same bin w.p. $d/\log^3 n$ in $h_{k+1}$.

For each leaf, its height is at least $b$ if each bin in its choices has height at least $b-1$, which happens with probability at most
\[ \left(\frac{{1.01 \cdot \log^3 n \cdot \frac{m}{n} \choose b-1}}{(\log^3 n/d)^{b-1}}\right)^d \le \left(\frac{(1.01 d \cdot \frac{m}{n})^{b-1}}{(b-1)!}\right)^d \le 2^{-3d^2} \cdot (\frac{n}{m})^{2d}
\]
from the proof of Theorem~\ref{thm:sym}.

Because these two types of events are on disjoint subsets of balls, the probability that any possible asymmetric witness tree in $\mathcal{F}_T$ exists in $t$-wise independent distributions over the suffixes is at most
\begin{align*}
|\mathcal{F}_T| \cdot \left(\frac{d}{\log^3 n}\right)^{|C|-1} \cdot \left(2^{-3d^2} \cdot (\frac{n}{m})^{2d}\right)^{\frac{(d-1)(|C|-3c)}{d}} \le & m \cdot \left(1.01 d \cdot \frac{m}{n} \right)^{|C|} \cdot \left(2^{-3d^2} \cdot (\frac{n}{m})^{2d}\right)^{|C|/3} \\ \le & m \cdot 2^{-f(l+1,1)} \le n^{-c-1}.
\end{align*}

We choose $\delta_2=n^{-c-1} \cdot (\log^3 n/d)^{-t}/|\mathcal{F}_T|=(\log n)^{-O(\log n)}$ such that in $\delta_2$-biased spaces, any possible asymmetric witness tree in $\mathcal{F}_T$ exists $h_{k+1}$ is at most happens with probability at most $n^{-c-1} + |\mathcal{F}_T| \cdot \delta_2 \cdot (\log^3/d)^{b d \cdot f(l+3c+1,1)} \le 2 n^{-c-1}$. At the same time, the number of possible configurations is at most $(f(l+3c+1,1)^2)^{3c} \le 0.1 n$.

From all discussion above, with probability at most $n^{-c}$, there exists a ball in the \emph{Always-Go-Left} mechanism with height at least $l+b+3c+1=\frac{\log n \log n}{d \log \phi_d} + O(1)$.
\end{proof}


\section{Heavy load}\label{sec:heavy_load}
We consider the derandomization of the 1-choice scheme when we have $m=n \cdot \poly (\log n)$ balls and $n$ bins. From the Chernoff bound, w.h.p, the max-load among $n$ bins is $\frac{m}{n} \left(1 + O(\sqrt{\log n} \cdot \sqrt{\frac{n}{m}}) \right)$ when we throw $m>n \log n$ balls into $n$ bins independently at random.
We modify the hash function from \cite{CRSW} with proper parameters for $m=\poly(\log n) \cdot n$ balls and prove the max-load is still $\frac{m}{n} \left(1 + O(\sqrt{\log n} \cdot \sqrt{\frac{n}{m}}) \right)$. We assume $m=\log^{a} n \cdot n$ for a constant $a \ge 1$ in the rest of this section.

\begin{theorem}\label{thm:heavy_load}
For any constant $c>0$ and $a \ge 1$, there exist a constant $C$ and a hash function from $U$ to $[n]$ generated by $O(\log n \log \log n)$ random bits such that for any $m=\log^a n \cdot n$ balls, with probability at least $1-n^{-c}$, the max-load of the $n$ bins in the 1-choice scheme with the hash function $h$ is at most $\frac{m}{n}\left(1+ C \cdot \sqrt{\log n} \cdot \sqrt{\frac{n}{m}}\right)$.
\end{theorem}

We omit $g$ in this section and change $h_1,\ldots,h_{k+1}$ with different parameters. We choose $k=\log \frac{\log n}{ (2 a) \log \log n}$, $h_{i}$ to denote a hash function from $U$ to $[n^{2^{-i}}]$ for $i \in [k]$, and $h_{k+1}$ to denote a hash function from $U$ to $[n^{2^{-k}}]=[\log^{2a} n]$ such that $h_1 \circ h_2 \circ \cdots \circ h_k \circ h_{k+1}$ constitute a hash function from $U$ to $[n]$. We set $\beta= 4(c+2) \sqrt{\log n} \sqrt{\frac{n}{m}}$. For convenience, we still think $h_1 \circ h_2 \circ \cdots \circ h_i$ as a hash function maps to $n^{1-2^{-i}}$ bins for any $i \le k$. In this section, we still use $\delta_1$-biased spaces on $h_1,\ldots,h_k$ and a $\delta_2$-biased space on $h_{k+1}$ for $\delta_1=1/\poly(n)$ and $\delta_2=(\log n)^{-O(\log n)}$.

\begin{claim}\label{clm:heavy_load}
For any constant $c>0$, there exists $\delta_1=1/\poly(n)$ such that given $m=\log^a n \cdot n$ balls, with probability $1 - n^{-c-1}$, for any $i \in [k]$ and any bin $b \in [n^{1-2^{-i}}]$, there are less than $\prod_{j \le i}(1+\frac{\beta}{(k+2-i)^2}) \cdot \frac{m}{n} \cdot n^{2^{-i}}$ balls in this bin.
\end{claim}

\begin{proof}
We still use induction on $i$. The base case is $i=0$. Because there are at most $m$ balls, the hypothesis is true.

Suppose it is true for $i=l$. Now we fix a bin and assume there are $s=\prod_{j \le l}(1+\frac{\beta}{(k+2-i)^2}) \cdot \frac{m}{n} n^{2^{-l}} \le (1+\beta) \frac{m}{n} n^{2^{-l}}$ balls in this bin from the induction hypothesis. $h_{l+1}$ maps these $s$ balls to $t=n^{2^{-(l+1)}}$ bins. We will prove that with high probability, every bin in these $t$ bins of $h_{l+1}$ contains at most $(1+\frac{\beta}{(k+1-l)^2})s/t$ balls.

We use $X_i \in \{0,1\}$ to denote whether ball $i$ is in one fixed bin of $[t]$ or not. Hence $\Pr[X_i=1]=1/t$. Let $Y_i=X_i - \E[X_i]$. Therefore $\E[Y_i]=0$ and $\E[|Y_i|^l] \le 1/t$ for any $l \ge 2$. Let $b=\beta 2^l$ for a large constant $\beta$ later.
\begin{align*}
\Pr_{D_{\delta_1}}[\sum_i X_i > (1+\frac{\beta}{(k+1-l)^2}) s/t] & \le \frac{\E_{D_{\delta_1}}[(\sum_i Y_i)^b]}{(\frac{\beta}{(k+1-l)^2} s/t)^b} \\
& \le \frac{ \sum_{i_1,\ldots,i_b} \E_U[Y_{i_1} \cdots Y_{i_b}] + \delta_1 s^{2b}}{(\frac{\beta}{(k+1-l)^2} s/t)^b}\\
& \le \frac{2^b b! (s/t)^{b/2} + \delta_1 s^{2b}}{(\frac{\beta}{(k+1-l)^2} s/t)^b}\\
& \le \left(\frac{2 b (s/t)}{(\frac{\beta}{(k+1-l)^2} s/t)^2}\right)^{b/2} + \delta_1 \cdot s^{2b}\\
\end{align*}
We use these bounds $k=\log \frac{\log n}{ (2 l) \log \log n} < \log \log n$, $b<\beta 2^k < \frac{\beta \log n}{(2 l) \log \log n}$ and $n^{2^{-l-1}} \ge n^{2^{k}} \ge \log^{2 l} n \ge (m/n)^2$ to simplify the above bound by
\begin{align*}
& \left(\frac{2 \log n }{\frac{\beta^2}{(\log \log n)^4} \cdot s/t}\right)^{b/2} + \delta_1 s^{2b}\\
\le & \left(\frac{2 \log^2 n}{ (\log n \cdot \frac{n}{m}) \cdot (\frac{m}{n} n^{2^{-l-1}})} \right)^{b/2} + \delta_1 s^{2b}\\
\le & \left( \frac{1}{n^{0.5 \cdot 2^{-l-1}}}\right)^{b/2} + \delta_1 s^{2b}\\
\le & n^{- 0.5 \cdot 2^{-l-1} \cdot \beta 2^l /2} + \delta_1 \left(\frac{2m}{n} n^{2^{-l}}\right)^{2 \beta 2^l}\le n^{ - \beta /8} + \delta_1 \cdot n^{6 \beta}.
\end{align*}
Hence we choose the two parameters $\beta > 8(c+2)$ and $\delta_1=n^{-6\beta - c - 2}$ such that the above probability is bounded by $2 n^{-c-2}$. Finally, we apply the union bound on $i$ and all bins.
\end{proof}

\begin{proofof}{Theorem \ref{thm:heavy_load}}
We first apply Claim \ref{clm:heavy_load} to $h_1,\ldots,h_k$.

In $h_{k+1}$, we first consider it as a $b=16(c+2)^2 \log n=O(\log n)$-wise independent distribution that maps $s<\prod_{j \le k}(1+\frac{\beta}{(k+2-i)^2}) \cdot \frac{m}{n} n^{2^{-k}}$ balls to $t = n^{2^{-k}}$ bins. From Lemma \ref{lem:k_wise_independence_chernoff} and Theorem 5 (I) in \cite{SSS95}, we bound the probability that one bin receives more than $(1+\beta) s/t$ by $e^{\beta^2 \cdot \E[s/t]/3} \le n^{-c-2}$ given $b \ge \beta^2 \E[s/t]$.

Then we choose $\delta_2=(\log n)^{- b \cdot 5 a}=(\log n)^{- O( \log n)}$ such that any $\delta_2$-biased space from $[2 \frac{m}{n} \log ^{2 a} n]$ to $[\log^{2 a} n]$ is $\delta_2 \cdot {2 \frac{m}{n} \log ^{2 a} n \choose \le b} \cdot (\log^{2 a} n)^{b}<n^{-c-2}$-close to a $b$-wise independent distribution. Hence in $h_{k+1}$, with probability at most $2 \cdot n^{-c-2}$, there is one bin that receives more than $(1+\beta) s/t$ balls. Overall, the number of balls in any bin of $[n]$ is at most \[ \prod_{i \le k} (1+\frac{\beta}{(k+2-i)^2}) (1+\beta) \frac{m}{n} \le (1 + \sum_{i \le k+1} \frac{\beta}{(k+2-i)^2}) \frac{m}{n} \le \left(1 + 2\beta \right) \frac{m}{n}.\]
\end{proofof}


\section*{Acknowledgement}
The author is grateful to David Zuckerman for his constant support and encouragement, as well as for many fruitful discussions. We thank Eric Price for introducing us to the multiple-choice schemes. We also thank the anonymous
referee for the detailed feedback and comments.

\bibliographystyle{alpha}
\bibliography{rand}

\appendix
\section{Proof of Lemma~\ref{lem:uniform_distributed}}\label{app:proof}
We apply an induction from $i=0$ to $i=k$. The base case $i=0$ is true because there are at most $m$ balls.

Suppose it is true for $i=l<d$. For a fixed bin $j \in [n^{1-\frac{1}{2^l}}]$, there are at most $(1+\beta)^l n^{\frac{1}{2^l}} \cdot \frac{m}{n}$ balls. With out loss of generality, we assume there are exactly $s=(1+\beta)^l n^{\frac{1}{2^l}} \cdot \frac{m}{n}$ balls from the induction hypothesis. Under the hash function $h_{l+1}$, we allocate these balls into $t=n^{\frac{1}{2^{l+1}}}$ bins.

We fix one bin in $h_{l+1}$ and prove that this bin receives at most \[(1+\beta)s/t=(1+\beta) \cdot (1+\beta)^l n^{\frac{1}{2^l}}/n^{\frac{1}{2^{l+1}}} \cdot \frac{m}{n}= (1+\beta)^{l+1} n^{\frac{1}{2^{l+1}}} \cdot \frac{m}{n}\]
balls with probability $\le 2 n^{-c-2}$ in a $\log^3 n$-wise $\delta_1$-biased space for $\delta_1=1/\poly(n)$ with a sufficiently large polynomial. We use $X_i \in \{0,1\}$ to denote the $i$th ball is in the bin or not. Hence $\E[\sum_{i \in [s]}X_i]=s/t$.

For convenience, we use $Y_i=X_i - \E[X_i]$. Hence $Y_i=1 - 1/t$ w.p. $1/t$, o.w. $Y_i = -1/t$. Notice that $\E[Y_i]=0$ and $|\E[Y_i^l]| \le 1/t$ for any $l \ge 2$.

We choose $b=2^l \cdot \beta=O(\log n)$ for a large even number $\beta$ and compute the $b$th moment of $\sum_{i \in [s]} Y_i$ as follows.
\begin{align*}
\Pr[|\sum_{i \in [s]} X_i| > (1+\beta) s/t] &\le \E_{\delta_1 \textit{-biased}}[(\sum_i Y_i)^b]/(\beta s/t)^b\\
& \le \frac{\E_{U}[(\sum_i Y_i)^b] + \delta_1 \cdot s^{2b} t^{b}}{(\beta s/t)^b} \\
& \le \frac{\sum_{i_1,\ldots,i_b} \E[Y_{i_1}\cdot Y_{i_2} \cdot \cdots \cdot Y_{i_b}]+ \delta_1 \cdot s^{3b}}{(\beta s/t)^b}\\
& \le \frac{\sum_{j=1}^{b/2}{b-j-1 \choose j-1} \cdot \frac{b!}{2^{b/2}} \cdot s^j (1/t)^j + \delta_1 \cdot s^{3b}}{(\beta s/t)^b}\\
& \le \frac{2^{b/2} b! \cdot (s/t)^{b/2} + \delta_1 \cdot s^{3b}}{(\beta s/t)^b}
\end{align*}
Because $s/t \ge n^{\frac{1}{2^{k}}} \ge \log^3 n, b \le \beta 2^k \le \frac{\beta \log n}{3\log \log n} \le (s/t)^{1/3}$ and $\beta=(\log n)^{-0.2}<(s/t)^{0.1}$, we simplify it to
\begin{align*}
& \left(\frac{2 b^2 \cdot s/t}{(\beta s/t)^2}\right)^{b/2} +\delta_1 \cdot s^{3b} \le \left(\frac{2 (s/t)^{2/3} \cdot s/t}{(s/t)^{1.8}}\right)^{b/2} + \delta_1 \cdot s^{3b} \\
= & (s/t)^{(-0.1) \cdot b/2} + \delta_1 \cdot s^{3b}
\le (n^{\frac{2}{2^{l+1}}})^{-0.1 \cdot \beta 2^l} + \delta_1 (n^{\frac{3}{2^l}})^{3 \beta \cdot 2^l} = n^{-c-2} + \delta_1 n^{9\beta} \le 2 n^{-c-2}.
\end{align*}
Finally we choose $\beta=40(c+2)=O(1)$ and $\delta_1=n^{-9\beta - c -2}$ to finish the proof.

\end{document}